\documentclass[preprint,a4paper,11pt,sort]{elsarticle}
\usepackage{fullpage}

\usepackage{amsmath,amssymb,latexsym}
\usepackage{amsthm} 
\usepackage{thm-restate}
\usepackage{tcolorbox}

\usepackage{xcolor}
\definecolor{darkred}{rgb}{0.75,0,0}

\usepackage{graphicx}

\usepackage{thmtools} 

\usepackage[colorlinks=true]{hyperref}

\usepackage[capitalize,nameinlink]{cleveref}
\newtheorem{theorem}{Theorem}[section]
\newtheorem{lemma}[theorem]{Lemma}
\crefname{lemma}{Lemma}{Lemmas}
\newtheorem{claim}{Claim}[theorem]
\crefname{claim}{Claim}{Claims}

\crefname{corollary}{Corollary}{Corollaries}

\crefname{proposition}{Proposition}{Propositions}

\theoremstyle{definition}
\newtheorem{observation}[theorem]{Observation}
\crefname{observation}{Observation}{Observations}

\newenvironment{subproof}[1][\proofname]{%
  \begin{proof}[#1]%
}{%
  \end{proof}%
}

\newcommand{\BBR}{\textsc{Balanced Biclique Reconfiguration}}
\newcommand{\CCR}{\textsc{Connected Components Reconfiguration}}
\newcommand{\ISR}{\textsc{Independent Set Reconfiguration}}
\newcommand{\CLR}{\textsc{Clique Reconfiguration}}

\newcommand{\cc}{\mathsf{cc}}
\newcommand{\mset}[1]{\{\mskip-5mu\{ #1 \}\mskip-5mu\}}

\usepackage{mathtools}

\usepackage[mathlines]{lineno}
\newcommand*\patchAmsMathEnvironmentForLineno[1]{
  \expandafter\let\csname old#1\expandafter\endcsname\csname #1\endcsname
  \expandafter\let\csname oldend#1\expandafter\endcsname\csname end#1\endcsname
  \renewenvironment{#1}
     {\linenomath\csname old#1\endcsname}
     {\csname oldend#1\endcsname\endlinenomath}}
\newcommand*\patchBothAmsMathEnvironmentsForLineno[1]{
  \patchAmsMathEnvironmentForLineno{#1}
  \patchAmsMathEnvironmentForLineno{#1*}}
\AtBeginDocument{
\patchBothAmsMathEnvironmentsForLineno{equation}
\patchBothAmsMathEnvironmentsForLineno{align}
\patchBothAmsMathEnvironmentsForLineno{flalign}
\patchBothAmsMathEnvironmentsForLineno{alignat}
\patchBothAmsMathEnvironmentsForLineno{gather}
\patchBothAmsMathEnvironmentsForLineno{multline}
}

\begin{document}

\title{Biclique Reconfiguration in Bipartite Graphs\tnoteref{t1}}
\tnotetext[t1]{%
  Partially supported by JSPS KAKENHI Grant Numbers
  JP22H00513, 
  JP24H00697, 
  JP25K03076, 
  JP25K03077. 
} 

\author[nagoya-u]{Yota Otachi}
\ead{otachi@nagoya-u.jp}

\author[nagoya-u]{Emi Toyoda}
\ead{toyoda.emi@nagoya-u.jp}

\affiliation[nagoya-u]{
    organization={Nagoya University},
    city={Nagoya},
    country={Japan}
}

\begin{abstract}
We prove that {\BBR} on bipartite graphs is PSPACE-complete.
This implies the PSPACE-completeness of the spanning variant of \textsc{Subgraph Reconfiguration} under the token jumping rule
for the property ``a graph is an $(i, j)$-complete bipartite graph,'' which was previously known only to be NP-hard [Hanaka et al.~TCS 2020].
Using our result, we also show that {\CCR} with two connected components is PSPACE-complete under all previously studied rules,
resolving an open problem of Nakahata~[COCOON 2025] in the negative.
\end{abstract}

\begin{keyword}
combinatorial reconfiguration, biclique, token jumping, connected components reconfiguration.
\end{keyword}

\maketitle


\section{Introduction}

In the framework of combinatorial reconfiguration on graphs,
we are given two substructures of a graph (e.g., vertex subsets) with a certain property
and are asked whether one can transform one into the other via step-by-step modifications, while preserving the property in the intermediate states.
In the literature, reconfiguration problems of several important structures have been studied extensively;
e.g., independent sets, vertex covers, cliques, dominating sets, matchings, and so on.
See the surveys~\cite{Heuvel13,Nishimura18} and the references therein.

In this paper, we study the reconfiguration problem of bicliques on bipartite graphs under the token jumping model.
A \emph{biclique}, or more specifically a \emph{$(p,q)$-biclique}, is a set of vertices that induces a complete bipartite graph $K_{p,q}$.
We call a $(p,p)$-biclique a \emph{balanced biclique}.
Bicliques, especially those in bipartite graphs, are natural and important objects with theoretical and practical applications~\cite{Lin18}.
Our main result is to show that {\BBR} on bipartite graphs, which we formally define later, is PSPACE-complete.
Previously, only the NP-hardness was known on bipartite graphs, while it was already known to be PSPACE-complete on general graphs~\cite{HanakaIMMNSSV20}.
As a theoretical application of our result,
we show the PSPACE-completeness of {\CCR} with two connected components by presenting a reduction from {\BBR} on bipartite graphs.
This answers an open problem of Nakahata~\cite{Nakahata25} in the negative.

\subsection{Confronting bipartiteness}
Our main technical contribution is to overcome the restrictions imposed by bipartiteness.
On general graphs, problems involving bicliques often admit straightforward reductions from the corresponding problems for independent sets:
for example, a graph $G$ contains an independent set of size $p$ if and only if 
$G * I_{p}$ contains a $(p,p)$-biclique, where $G * I_{p}$ is the graph obtained from $G$ by adding $p$ vertices adjacent to all vertices of $G$.
However, arguably the most natural restriction of asking the input graphs to be bipartite makes hardness proofs considerably more challenging.
Note that $G * I_{p}$ is non-bipartite unless $G$ is edgeless.

For general graphs, Hanaka et al.~\cite[Theorem~12]{HanakaIMMNSSV20} showed that {\BBR} is PSPACE-complete by presenting a reduction from \textsc{Maximum} {\ISR},
which is basically the same $G * I_{p}$-reduction above (with an extra trick) and thus cannot be used for the bipartite case.
For bipartite graphs, they presented a reduction~\cite[Theorem~13]{HanakaIMMNSSV20}
from the problem of finding a large balanced biclique on bipartite graphs, which shows only NP-hardness.
To be more precise, their target problem there was
\emph{the spanning variant of \textsc{Subgraph Reconfiguration} under the TJ rule for the property ``a graph is an $(i, j)$-complete bipartite graph,''}
which is the reconfiguration problem of vertex sets that induce graphs that contain $K_{i,j}$ as a spanning subgraph.
On bipartite graphs, this problem with $i=j$ is equivalent to {\BBR}, and their NP-hardness result was shown for this case.
Our result in this paper implies the PSPACE-completeness of their problem as well.

Two further examples are the classical and parameterized complexity of finding a large balanced biclique in bipartite graphs.
(Note that, for general graphs, both cases admit the $G * I_{p}$ reduction.)
In their classic book~\cite[GT24]{GareyJ79},
Garey and Johnson stated without proof that the problem of finding a large balanced biclique on bipartite graphs is NP-complete. 
Later, in the 20th NP-completeness column~\cite{Johnson87}, Johnson presented a full proof and remarked that:
\begin{quote}
  ``\textit{Although a hint is provided in [G\&J] (``transformation from CLIQUE''),
  the NP-completeness proof seems to have been sufficiently subtle to have eluded many who tried to regenerate it.}''
\end{quote}
In the field of parameterized complexity,
the complexity of the problem of finding a large balanced biclique on bipartite graphs 
parameterized by the solution size was a long-standing open problem.
Downey and Fellows~\cite[\S 33.1]{DowneyF13} listed this problem as one of the most infamous open problems.\footnote{%
To be precise, the problem that they mentioned is a slightly more general problem of finding a large balanced complete bipartite subgraph on general graphs,
which is equivalent when input graphs are bipartite.}
There, it is said that
\begin{quote}
  ``\textit{Almost everyone considers that this problem should obviously be W[1]-hard, 
    and almost everyone quickly finds the same easy (but erroneous) proof. 
    It is rather an embarrassment to the field that the question remains open after all these years!}''
\end{quote}
This problem was finally resolved by Lin~\cite{Lin18}, who proved its W[1]-hardness.
This result received both the best paper award and the best student paper award at SODA 2015~\cite{Lin15}.

In the context of combinatorial reconfiguration, there is another example (not about bicliques) 
that actually shows a case where the bipartiteness restriction cannot be overcome.
{\ISR} is one of the most well-studied problems in the field of combinatorial reconfiguration.
It is studied under three modification models,
\emph{token addition and removal} (TAR)~\cite{ItoDHPSUU11}, \emph{token jumping} (TJ)~\cite{KaminskiMM12}, and \emph{token sliding} (TS)~\cite{HearnD05},
and known to be PSPACE-complete under all of them on general graphs.
The complexity of {\ISR} on bipartite graphs was open for some years,
and Lokshtanov and Mouawad \cite{LokshtanovM19} resolved it in a somewhat surprising way:
they showed that on bipartite graphs, it is PSPACE-complete under TS, but NP-complete under TJ and TAR\@.
Thus, under some models, restriction to bipartite graphs makes the problem strictly easier than on general graphs (assuming that $\mathrm{NP} \ne \mathrm{PSPACE}$).

\subsection{Our results}

\paragraph{The main result}
As mentioned above, our main result is to completely determine the complexity of {\BBR} on bipartite graphs by showing that it is PSPACE-complete.
We now formally define the problem.
Let $G = (V,E)$ be a graph.
A vertex set $S \subseteq V$ is a \emph{$(p,q)$-biclique} if $G[S]$ is isomorphic to the complete bipartite graph $K_{p,q}$.
A sequence $S_{0}, \dots, S_{\ell}$ of vertex sets is a \emph{TJ-sequence of $(p,q)$-bicliques}
if each $S_{i}$ is a $(p,q)$-biclique and $|S_{i} \setminus S_{i+1}| = |S_{i+1} \setminus S_{i}| = 1$ for $0 \le i \le \ell-1$.
We call such a move from $S_{i}$ to $S_{i+1}$ a \emph{TJ-move}.

Now we can formulate the problem {\BBR} as follows.
\begin{tcolorbox}
\begin{description}
  \setlength{\itemsep}{0pt}
  \item[Problem.] {\BBR}
  \item[Input.] A graph $G$ and two $(p,p)$-bicliques $S$ and $S'$ of $G$.
  \item[Question.] Is there a TJ-sequence of $(p,p)$-bicliques from $S$ to $S'$ in $G$?
\end{description}
\end{tcolorbox}

Our main result is formalized as follows.
\begin{restatable}{theorem}{thmBBR}
\label{thm:bbr-hardness}
{\BBR} is $\mathrm{PSPACE}$-complete on bipartite graphs.
\end{restatable}

We prove \cref{thm:bbr-hardness} in \cref{sec:bbr-hardness}.
There, we present a reduction from {\CLR}~\cite{ItoDHPSUU11,ItoOO23} to {\BBR}.
We first follow the reduction of Johnson~\cite{Johnson87} from the problem of finding a large clique to that of finding a large balanced biclique.
Although this already gives a correspondence between cliques in a general graph and balanced bicliques in a bipartite graph, there remain two obstacles: 
\begin{enumerate}
  \setlength{\itemsep}{0pt}
  \item the obtained bicliques are \emph{locked}; i.e., no TJ-moves can be applied to them;
  \item the bicliques corresponding to two TJ-adjacent cliques differ in a large number of vertices.
\end{enumerate}
At first glance, the first issue looks severe. However, there is a trick to \emph{unlock} the bicliques by reducing their size by~$1$.
This is a simple but crucial step, and showing its correctness is the most nontrivial part of the proof.
The second issue is easier to handle by applying a rather standard technique of simulating a single step in the original problem ({\CLR})
by several steps in the target problem ({\BBR}).

\paragraph{The second result: an application of the main result}
As a theoretical application of the main result, we obtain the following theorem,
which will be proved in \cref{section:ccr-hardness}.
\begin{restatable}{theorem}{thmCCR}
\label{thm:ccr-hardness}
{\CCR} with two connected components of the same size is PSPACE-complete
under both CJ and CS rules even on co-bipartite graphs.
\end{restatable}

We defer the formal definitions for \cref{thm:ccr-hardness} to \cref{section:ccr-hardness}.
Roughly speaking, {\CCR} asks to reconfigure a vertex set that induces a graph with specific sizes of connected components.
It is a generalization of {\ISR} as it can ask all connected components to have size~$1$.
Nakahata~\cite{Nakahata25} introduced {\CCR}, showed several positive and negative results,
and asked whether {\CCR} admits a polynomial-time algorithm when the number of connected components in the reconfigured set is a constant.
\cref{thm:ccr-hardness} answers this question in the negative assuming $\mathrm{P} \ne \mathrm{PSPACE}$.


\section{Preliminaries}

Let $G = (V,E)$ be a graph. 
For $S \subseteq V$, we denote by $G[S]$ the subgraph of $G$ induced by $S$,
and by $E(S) = \{\{u,v\} \in E \mid u, v \in S\}$ its edge set.

A \emph{clique} is a set of pairwise adjacent vertices and a \emph{$k$-clique} is a clique of size $k$.
An \emph{independent set} is a set of pairwise non-adjacent vertices.
A graph $G = (V, E)$ is \emph{bipartite} if its vertex set $V$ can be partitioned into two independent sets $A$ and $B$.
To emphasize that $G$ is bipartite with a bipartition $\{A,B\}$, we sometimes write $G = (A,B;E)$.
Note that if $G$ is connected, its bipartition is unique.

The \emph{complement} of a graph $G = (V,E)$ is $\bar{G} = (V, \{\{u,v\} \mid u,v \in V, \; u \ne v, \; \{u,v\} \notin E \})$.
The complement of a bipartite graph is a \emph{co-bipartite} graph.
The \emph{bipartite complement} of a bipartite graph $H = (A, B; E)$ is the bipartite graph with the same vertex set $A \cup B$
and the edge set $\{\{a,b\} \mid a \in A, \; b \in B, \; \{a,b\} \notin E\}$.

Let $S_{0}, \dots, S_{\ell}$ be a TJ-sequence of $(p,q)$-bicliques in a connected bipartite graph $G = (A,B;E)$ such that $|S_{0} \cap A| = p$ and $|S_{0} \cap B| = q$.
Observe that if $|p-q| \ne 1$, then $|S_{i} \cap A| = p$ and $|S_{i} \cap B| = q$ hold for every $i$.
This implies that every TJ-move involves vertices only in $A$ or only in $B$; that is,
either $(S_{i} \setminus S_{i+1}) \cup (S_{i+1} \setminus S_{i}) \subseteq A$ or 
$(S_{i} \setminus S_{i+1}) \cup (S_{i+1} \setminus S_{i}) \subseteq B$ for $0 \le i \le \ell-1$.
In our reductions, we consider the cases where $|p-q| \ne 1$ and thus have such \emph{side-preserving} TJ-moves only.

To simplify the presentation of reconfiguration steps,
we sometimes write $S - v$, $S + v$, and $S - u + v$
instead of $S \setminus \{v\}$, $S \cup \{v\}$, and $(S \setminus \{u\}) \cup \{v\}$, respectively.


\section{{\BBR} on bipartite graphs}
\label{sec:bbr-hardness}

This section is devoted to the proof of \cref{thm:bbr-hardness}, which is restated below.

\thmBBR*

Since the membership of {\BBR} in PSPACE is known~\cite{HanakaIMMNSSV20}, it suffices to show the PSPACE-hardness.
To this end, we present a reduction from {\CLR} on general graphs to {\BBR} on bipartite graphs.

Let $G = (V,E)$ be a graph. A sequence $K_{0}, \dots, K_{\ell}$ is a \emph{TJ-sequence of $k$-cliques}
if $K_{i}$ is a $k$-clique for $0 \le i \le \ell$
and $|K_{i} \setminus K_{i+1}| = |K_{i+1} \setminus K_{i}| = 1$ for $0 \le i \le \ell-1$.
Given a graph $G$ and $k$-cliques $K$ and $K'$ of $G$,
{\CLR} asks whether there exists a TJ-sequence of $k$-cliques from $K$ to~$K'$.

{\ISR} under TJ is defined almost the same way as {\CLR}, where the only difference is
that $k$-cliques in the definition are replaced with independent sets of size $k$.
{\ISR} under TS additionally requires that the swapped vertices in each step are adjacent.

It is known that {\CLR} is PSPACE-complete~\cite{ItoDHPSUU11,ItoOO23}.
For our purpose, we need the following hardness with an additional condition on $k$ and $n = |V|$.
\begin{lemma}
\label{lem:clr-n-vs-k}
{\CLR} of $k$-cliques on $n$-vertex graphs is PSPACE-complete even if $n-k-1 = \binom{k}{2}$.
\end{lemma}
\begin{proof}
It is known that {\ISR} under TS is PSPACE-complete even if the vertex set of the input graph is partitioned into cliques of size at most~$3$
and the input independent sets intersect all these cliques~\cite{HearnD05} (see also~\cite[Section 3.2]{BonsmaC09}).
Observe that such independent sets are maximum and their size $k$ is at least~$n/3$.
Since TS and TJ are equivalent for maximum independent sets, {\ISR} under TJ is also PSPACE-complete if $k \ge n/3$.
Furthermore, since {\ISR} under TJ on $G$ is equivalent to {\CLR} on its complement $\bar{G}$,
it follows that {\CLR} is PSPACE-complete even if $k \ge n/3$.

Let $\langle G, K, K' \rangle$ be an instance of {\CLR} with $n = |V(G)|$ and $k = |K| = |K'|$ such that $k \ge n / 3$.
We assume that $n > 12$ since otherwise the problem can be solved in polynomial time.
A simple calculation under the assumptions $k \ge n / 3$ and $n > 12$ yields $n - k-1 \le \binom{k}{2}$.
If $n - k - 1 = \binom{k}{2}$, we are done.
Otherwise, we obtain a new graph $G'$ from $G$ by adding $\binom{k}{2} - (n - k - 1)$ isolated vertices.
Clearly, $\langle G', K, K' \rangle$ and $\langle G, K, K' \rangle$ are equivalent instances of {\CLR}
and $G'$ satisfies $n'-k-1 = \binom{k}{2}$, where $n' = |V(G')|$ ($= \binom{k}{2} + k + 1$).
\end{proof}

\paragraph{The bipartite graph $H$}
In the rest of this section, we fix an instance $\langle G, K, K'\rangle$ of {\CLR}
such that $G = (V,E)$ has $n$ vertices and $|K|=|K'|=k$.
By \cref{lem:clr-n-vs-k}, we may assume that $n-k-1 = \binom{k}{2}$.
We assume that $k \ge 3$ (and thus $n = k+1+ \binom{k}{2} \ge 7$), since otherwise the problem can be solved in polynomial time.
This implies that $G$ contains a triangle and hence, it is non-bipartite.

Let $H$ be the bipartite complement of the \emph{incidence graph} of $G$ (see \cref{fig:reduction});
that is, 
\[
  H = (V, E; \, \{\{v,e\} \mid v \in V, \, e \in E, \, v \notin e\}).
\]

\begin{figure}[tb]
  \centering
  \includegraphics{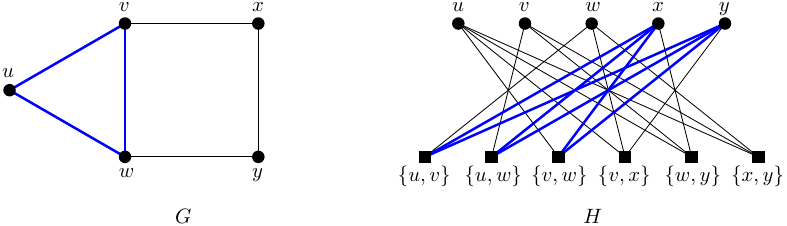}
  \caption{The construction for \cref{thm:bbr-hardness} (without the condition $n-k-1 = \binom{k}{2}$). 
    A clique $K \subseteq V$ of $G = (V,E)$ corresponds to the biclique $(V \setminus K) \cup E(K)$.
    For example, $\{u,v,w\}$ corresponds to $\{x,y\} \cup \{\{u,v\}, \{u,w\}, \{v,w\}\}$.}
  \label{fig:reduction}
\end{figure}

\begin{observation}
$H$ is connected (and thus, its bipartition $\{V,E\}$ is unique).
\end{observation}
\begin{proof}
Since each $e \in E$ is adjacent to $n-2$ elements of $V$ and since $n-2 > n/2$ by the assumption $n \ge 7$, any two vertices $e, e' \in E$ of $H$ have a common neighbor in $V$.
Hence, all elements of $E$ belong to the same connected component of $H$.
Since $G$ is non-bipartite, no vertex of $G$ is incident to all edges of $G$ (otherwise $G$ would be a subgraph of a star).
This implies that each $v \in V$ is adjacent to at least one element of $E$ in $H$. Thus, $H$ is connected.
\end{proof}

The following lemma follows directly from the discussion by Johnson~\cite{Johnson87}.
We include a proof here to be self-contained.
\begin{lemma}
[\cite{Johnson87}]
\label{lem:clique-biclique}
For non-empty sets $U \subseteq V$ and $F \subseteq E$,
$(V \setminus U) \cup F$ is a biclique of $H$
if and only if
$F \subseteq E(U)$.
\end{lemma}
\begin{proof}
We first prove the only-if direction.
If an edge $e \in F$ does not belong to $E(U)$, 
then $e$ has a vertex $v \in V \setminus U$ as an endpoint, and hence $\{v, e\} \notin E(H)$.
Thus, $(V \setminus U) \cup F$ is not a biclique.
For the if direction, assume that $F \subseteq E(U)$.
Then, for all $e \in F$ and $v \in V \setminus U$,
we have $v \notin e$, and thus $\{v,e\} \in E(H)$.
Therefore, $(V \setminus U) \cup F$ is a biclique of $H$.
\end{proof}

Observe that, when $|F| = \binom{|U|}{2}$,
\cref{lem:clique-biclique} implies that 
$(V \setminus U) \cup F$ is a biclique of $H$
if and only if
$U$ is a clique of $G$.
See \cref{fig:reduction}.
Now, a seemingly promising approach would be to directly use this correspondence
for obtaining a TJ-sequence of bicliques from a TJ-sequence of cliques (and vice versa).
However, we can see that the biclique is \emph{locked} in the sense that $(V \setminus U) \cup F$ admits no TJ-moves in this case:
since $|F| = \binom{|U|}{2}$ and $U$ is a clique, we have $F = E(U)$;
this implies that, in $H$, each $v \in U$ has non-neighbors $\{\{u,v\} \mid u \in U \setminus \{v\}\} \subseteq F$;
and each $\{u,v\} \in E \setminus F$ has non-neighbors $u,v \in V \setminus U$.

Our trick to \emph{unlock} the biclique is to remove one arbitrary vertex in the $V$-side of the biclique.
Although this step may appear ad hoc, it is precisely what makes the reconfiguration possible.

\paragraph{The bicliques $S$ and $S'$}
Recall that $K$ and $K'$ are $k$-cliques of $G$ with $n-k-1 = \binom{k}{2}$.
Let $\hat{S} = (V \setminus K) \cup E(K)$ and $\hat{S}' = (V \setminus K') \cup E(K')$.
By \cref{lem:clique-biclique}, $\hat{S}$ and $\hat{S}'$ are $(n-k, \binom{k}{2})$-bicliques of $H$.
Let $S$ be the set obtained from $\hat{S}$ by removing an arbitrary vertex from the $V$-side $\hat{S} \cap V = V \setminus K$.
Similarly, let $S'$ be the set obtained from $\hat{S}'$ by removing an arbitrary vertex from the $V$-side $\hat{S}' \cap V = V \setminus K'$.
Let $p = n-k-1$ ($= \binom{k}{2}$). Then, $S$ and $S'$ are $(p,p)$-bicliques of $H$.
Since $H$ has the unique bipartition $\{V,E\}$, each TJ-move happens in the $V$-side or in the $E$-side.
In particular, the size of each side is always $p$.

In the following, we show that
there is a TJ-sequence of $k$-cliques from $K$ to $K'$ in $G$
if and only if
there is a TJ-sequence of $(p,p)$-bicliques from $S$ to $S'$ in $H$.

\begin{lemma}
If there is a TJ-sequence of $k$-cliques from $K$ to $K'$ in $G$,
then there is a TJ-sequence of $(p,p)$-bicliques from $S$ to $S'$ in $H$.
\end{lemma}
\begin{proof}
Let $K_{0}, K_{1}, \dots, K_{\ell}$ be a TJ-sequence of $k$-cliques from $K = K_{0}$ to $K' = K_{\ell}$ in $G$.
We define sequences $\hat{S}_{0}, \hat{S}_{1}, \dots, \hat{S}_{\ell}$ and $S_{0}, S_{1}, \dots, S_{\ell}$ of vertex sets of $H$ as follows.
\begin{itemize}
  \setlength{\itemsep}{0pt}
  \item Set $\hat{S}_{0} = \hat{S}$, $S_{0} = S$, $\hat{S}_{\ell} = \hat{S}'$, and $S_{\ell} = S'$.
  \item For $1 \le i \le \ell-1$, first set $\hat{S}_{i} = (V \setminus K_{i}) \cup E(K_{i})$,
  and then set $S_{i}$ to the set obtained from $\hat{S}_{i}$ by removing an arbitrary vertex from the $V$-side $V \setminus K_{i}$.
  \item For $0 \le i \le \ell$, let $w_{i}$ denote the element removed from $\hat{S}_{i}$ to obtain $S_{i}$; i.e., $S_{i} = \hat{S}_{i} - w_{i}$.
\end{itemize}
By \cref{lem:clique-biclique}, $\hat{S}_{i}$ is a $(p+1, p)$-biclique of $H$, and thus $S_{i}$ is a $(p, p)$-biclique of $H$.

To prove the lemma, it suffices to show that for $0 \le i \le \ell-1$,
there is a TJ-sequence of $(p,p)$-bicliques from $S_{i}$ to $S_{i+1}$ in $H$.
Fix $i \in \{0, 1, \dots, \ell-1\}$
and let $K_{i} \setminus K_{i+1} = \{u_{i}\}$ and $K_{i+1} \setminus K_{i} = \{v_{i+1}\}$.
We construct a TJ-sequence from $S_{i}$ to $S_{i+1}$ as follows.
\begin{enumerate}
  \item Remove $v_{i+1}$ and add $w_{i}$ (if $w_{i} \ne v_{i+1}$).

  \item For each $x \in (K_{i} \cup K_{i+1}) \setminus \{u_{i}, v_{i+1}\}$, 
  remove $\{u_{i}, x\}$ and add $\{v_{i+1}, x\}$.

  \item Remove $w_{i+1}$ and add $u_{i}$ (if $w_{i+1} \ne u_{i}$).
\end{enumerate}
We show the validity of this sequence below.

The set obtained by the first step is $S_{i} - v_{i+1} + w_{i} = (V \setminus (K_{i} + v_{i+1})) \cup E(K_{i})$,
which is a $(p,p)$-biclique by \cref{lem:clique-biclique}.
Since $K_{i} + v_{i+1} = K_{i} \cup K_{i+1}$,
it follows that $(V \setminus (K_{i} + v_{i+1})) \cup E(K_{i}) = (V \setminus (K_{i} \cup K_{i+1})) \cup E(K_{i})$.

During the second step, the $V$-side $V \setminus (K_{i} \cup K_{i+1})$ stays the same
and the $E$-side is always a subset of $E(K_{i}) \cup E(K_{i+1}) \subseteq E(K_{i} \cup K_{i+1})$.
By \cref{lem:clique-biclique}, every set appearing in this step is a $(p,p)$-biclique.
The set obtained after the second step is 
$(V \setminus (K_{i} \cup K_{i+1})) \cup E(K_{i+1}) = (V \setminus (K_{i+1} + u_{i})) \cup E(K_{i+1})$
as $K_{i} \cup K_{i+1} = K_{i+1} + u_{i}$.

The third step obtains $(V \setminus (K_{i+1} + w_{i+1})) \cup E(K_{i+1}) = S_{i+1}$.
\end{proof}

\begin{lemma}
If there is a TJ-sequence of $(p,p)$-bicliques from $S$ to $S'$ in $H$,
then there is a TJ-sequence of $k$-cliques from $K$ to $K'$ in $G$.
\end{lemma}
\begin{proof}
Let $S_{0}, \dots, S_{\ell}$ be a TJ-sequence of $(p,p)$-bicliques from $S = S_{0}$ to $S' = S_{\ell}$ in $H$.
For $0 \le i \le \ell$, let $V_{i} = V \setminus S_{i}$. Note that $|V_{i}| = |V \setminus S_{i}| = n - p = k+1$.

\begin{claim}
\label{clm:intersection-clique}
For $0 \le i \le \ell-1$,
if $V_{i} \ne V_{i+1}$, then $V_{i} \cap V_{i+1}$ is a $k$-clique of $G$. 
\end{claim}
\begin{subproof}
Since $V_{i} \ne V_{i+1}$, it holds that $V \cap S_{i} \ne V \cap S_{i+1}$, which implies that $E \cap S_{i} = E \cap S_{i+1}$.
Let $Y = E \cap S_{i}= E \cap S_{i+1}$.
Since $S_{i}$ and $S_{i+1}$ are bicliques of $H$,
\cref{lem:clique-biclique} implies that $Y \subseteq E(V_{i})$ and $Y \subseteq E(V_{i+1})$;
that is, $Y \subseteq E(V_{i}) \cap E(V_{i+1}) = E(V_{i} \cap V_{i+1})$.
Since $|V_{i} \cap V_{i+1}| = |V_{i}|-1 = k$ and $|Y| = p = \binom{k}{2}$, 
$V_{i} \cap V_{i+1}$ is a $k$-clique of $G$.
\end{subproof}

Now we define a sequence $K_{0}, \dots, K_{\ell+1}$ as follows:
\begin{itemize}
  \setlength{\itemsep}{0pt}
  \item $K_{0} = K$, $K_{\ell+1} = K'$;
  \item for $0 \le i \le \ell-1$, if $V_{i+1} = V_{i}$, then $K_{i+1} = K_{i}$; otherwise, $K_{i+1} = V_{i} \cap V_{i+1}$.
\end{itemize}
By \cref{clm:intersection-clique}, each $K_{i}$ is a $k$-clique of $G$.
Observe that, for $0 \le i \le \ell$, the set $K_{i}$ is a proper subset of $V_{i}$ with exactly one missing element; that is, $K_{i} = V_{i} \setminus \{w\}$ for some $w \in V_{i}$.
To prove the lemma, it suffices to show that either $K_{i} = K_{i+1}$ or $|K_{i} \setminus K_{i+1}| = |K_{i+1} \setminus K_{i}| = 1$.

First assume that $0 \le i \le \ell-1$.
Assume also that $V_{i+1} \ne V_{i}$ (and thus $|V_{i} \setminus V_{i+1}| = |V_{i+1} \setminus V_{i}| = 1$) since otherwise $K_{i} = K_{i+1}$.
Let $K_{i} = V_{i} \setminus \{w\}$ for some $w \in V_{i}$.
If $w \notin V_{i+1}$, then
\[
  K_{i+1} = V_{i} \cap V_{i+1} = (V_{i} \setminus \{w\}) \cap V_{i+1} = K_{i} \cap V_{i+1} = K_{i},
\]
where the last equality holds by $|K_{i+1}| = |K_{i}|$.
If $w \in V_{i+1}$, then
\begin{align*}
  |K_{i+1} \setminus K_{i}| 
  &= |(V_{i} \cap V_{i+1}) \setminus (V_{i} \setminus \{w\})|
  = |\{w\} \cap V_{i+1}|
  = |\{w\}| = 1,
  \\
  |K_{i} \setminus K_{i+1}|
  &= |(V_{i} \setminus \{w\}) \setminus (V_{i} \cap V_{i+1})|
  = |V_{i} \setminus (V_{i} \cap V_{i+1})|
  = |V_{i} \setminus V_{i+1}| = 1.
\end{align*}

Finally, we consider the case of $i = \ell$.
Since $S_{\ell}$ is a biclique of $H$ and $S_{\ell} \cap E = E(K') = E(K_{\ell+1})$, 
\cref{lem:clique-biclique} implies that $E(K_{\ell+1}) \subseteq E(V \setminus S_{\ell}) = E(V_{\ell})$, and thus $K_{\ell+1} \subseteq V_{\ell}$.
Since $K_{\ell} \subseteq V_{\ell}$ and $|K_{\ell}| = |K_{\ell+1}| = |V_{\ell}| - 1$, 
it follows that either $K_{\ell} = K_{\ell+1}$ or $|K_{\ell} \setminus K_{\ell+1}| = |K_{\ell+1} \setminus K_{\ell}| = 1$.
\end{proof}


\section{{\CCR} with two connected components}
\label{section:ccr-hardness}
Let $G = (V,E)$ be a graph.
For $S \subseteq V$, let $\cc_{G}(S)$ be the set of the vertex sets of connected components of $G[S]$
and let $\mu_{G}(S)$ be the multiset of the orders of connected components of $G[S]$;
that is, $\mu_{G}(S) = \mset{|X| : X \in \cc_{G}(S)}$.\footnote{%
We denote a multiset by double curly braces.}
Let $M$ be a multiset of positive integers. 
If $S \subseteq V$ satisfies $\mu_{G}(S) = M$, then we call $S$ an \emph{$M$-set}.

For {\CCR}, the two rules CJ (component jumping) and CS (component sliding) are studied.
A sequence $S_{0}, \dots, S_{\ell}$ of vertex sets is a \emph{CJ-sequence of $M$-sets} in $G$
if each $S_{i}$ is an $M$-set and $|\cc_{G}(S_{i}) \setminus \cc_{G}(S_{i+1})| = |\cc_{G}(S_{i+1}) \setminus \cc_{G}(S_{i})| = 1$ for $0 \le i \le \ell-1$.
We call such a move from $S_{i}$ to $S_{i+1}$ a \emph{CJ-move}.
A CJ-sequence $S_{0}, \dots, S_{\ell}$ of $M$-sets is a \emph{CS-sequence} of $M$-sets if it additionally satisfies the following condition for $0 \le i \le \ell-1$:
the unique elements $X_{i} \in \cc_{G}(S_{i}) \setminus \cc_{G}(S_{i+1})$ and 
$X_{i+1} \in \cc_{G}(S_{i+1}) \setminus \cc_{G}(S_{i})$ form a connected graph; that is, $G[X_{i} \cup X_{i+1}]$ is connected.
Furthermore, if it also satisfies $|X_{i} \setminus X_{i+1}| = |X_{i+1} \setminus X_{i}| = 1$ for $0 \le i \le \ell-1$, then the sequence is a \emph{CS1-sequence} of $M$-sets.
We define \emph{CS-moves} and \emph{CS1-moves} analogously.
It is known that 
there is a CS-sequence between two $M$-sets if and only if 
there is a CS1-sequence between them~\cite{Nakahata25}.

Now we are ready to define the problem and recall \cref{thm:ccr-hardness}.
\begin{tcolorbox}
\begin{description}
  \setlength{\itemsep}{0pt}
  \item[Problem.] {\CCR} under CJ (CS, resp.)
  \item[Input.] A graph $G$ and two $M$-sets $S$ and $S'$.
  \item[Question.] Is there a CJ-sequence (CS-sequence, resp.) of $M$-sets from $S$ to $S'$ in $G$?
\end{description}
\end{tcolorbox}

\thmCCR*

Since the membership in PSPACE is known~\cite{Nakahata25}, we only show the PSPACE-hardness 
by presenting a reduction from {\BBR} on bipartite graphs.

The next lemma is the key observation showing that the $(p,q)$-bicliques in a bipartite graph $G$
are exactly the $\mset{p,q}$-sets of its complement $\bar{G}$.
\begin{lemma}
\label{lem:bbr-ccr-complement-sets}
Let $G = (A,B;E)$ be a bipartite graph and $p,q$ be positive integers.
Then, $S$ is a $(p,q)$-biclique of $G$
if and only if 
$S$ is a $\mset{p,q}$-set of $\bar{G}$.
\end{lemma}
\begin{proof}
For the only-if direction, observe that 
if $S$ is a $(p,q)$-biclique of $G$ with bipartition $\{X, Y\}$,
 then $\cc_{\bar{G}}(S) = \{X, Y\}$, and thus $\mu_{\bar{G}}(S) = \mset{p,q}$.

To prove the if direction, assume that $\cc_{\bar{G}}(S) = \{X, Y\}$ and $(|X|, |Y|) = (p,q)$.
If $X$ intersects both cliques $A$ and $B$ of $\bar{G}$, then $X$ is a dominating set of $\bar{G}$, which contradicts that $\cc_{\bar{G}}(S) = \{X, Y\}$.
Thus $X$ cannot intersect both $A$ and $B$. By symmetry, we may assume that $X \subseteq A$.
Now, since $A$ is a clique of $\bar{G}$ and $\bar{G}$ has no edge between $X$ and $Y$, it follows that $Y \subseteq B$.
The discussion so far implies that, in the original graph $G$, $X$ and $Y$ are independent sets and there are all possible edges between them;
that is, $X \cup Y$ ($=S$) is a $(p, q)$-biclique of~$G$.
\end{proof}

\begin{lemma}
\label{lem:bbr-ccr-complement-moves}
Let $G = (A,B;E)$ be a bipartite graph and $T$ and $T'$ be $(p,p)$-bicliques of $G$.
There is a TJ-move from $T$ to $T'$ in $G$
if and only if
there is a CS1-move from $T$ to $T'$ in $\bar{G}$.
\end{lemma}
\begin{proof}
Assume that $T = X \cup Y$ and $T' = X' \cup Y'$ such that $X, X' \subseteq A$ and $Y, Y' \subseteq B$.
Since $T$ and $T'$ are bicliques of $G$, it follows that $\cc_{\bar{G}}(T) = \{X, Y\}$ and $\cc_{\bar{G}}(T') = \{X', Y'\}$.
Since $A$ and $B$ are cliques of $\bar{G}$, both $X \cup X'$ ($\subseteq A$) and $Y \cup Y'$ ($\subseteq B$) are connected in $\bar{G}$.
Since $|X| = |Y| = |X'| = |Y'|$ ($ = p$), 
we can see that $|(X \cup Y) \setminus (X' \cup Y')| = |(X' \cup Y') \setminus (X \cup Y)| = 1$
if and only if all conditions below are satisfied:
\begin{itemize}
  \setlength{\itemsep}{0pt}
  \item either $X \ne X'$ and $Y = Y'$, or, $X = X'$ and $Y \ne Y'$;
  \item if $X \ne X'$, then $|X \setminus X'| = |X' \setminus X| = 1$;
  \item if $Y \ne Y'$, then $|Y \setminus Y'| = |Y' \setminus Y| = 1$.
\end{itemize}
Therefore, 
there is a TJ-move from $T = X \cup Y$ to $T' = X' \cup Y'$ in $G$
if and only if
there is a CS1-move from $T$ to $T'$ in $\bar{G}$.
\end{proof}

\cref{lem:bbr-ccr-complement-sets,lem:bbr-ccr-complement-moves} imply that,
for a bipartite graph $G$,
a sequence $S_{0}, \dots, S_{\ell}$ is a TJ-sequence of $(p,p)$-bicliques in $G$
if and only if 
$S_{0}, \dots, S_{\ell}$ is a CS1-sequence of $\mset{p,p}$-sets in $\bar{G}$.
This proves \cref{thm:ccr-hardness} with CS since CS1 and CS are equivalent~\cite{Nakahata25}.

Finally, we prove \cref{thm:ccr-hardness} with CJ by showing the equivalence of CS and CJ in the current setting.
Let $\bar{G}$ be the complement of a bipartite graph $G = (A,B;E)$,
and $T$ and $T'$ be $\mset{p,p}$-sets in $\bar{G}$ with $X = T \cap A$, $Y = T \cap B$, $X' = T' \cap A$, and $Y' = T' \cap B$.
Assume that there is a CJ-move from $T$ to $T'$; that is, $|\{X,Y\} \setminus \{X',Y'\}| = |\{X',Y'\} \setminus \{X,Y\}| = 1$.
By symmetry, we may assume that $X \in \{X,Y\} \setminus \{X',Y'\}$.
If $Y' \in \{X',Y'\} \setminus \{X,Y\}$, then $\cc_{\bar{G}}(T') = \{Y, Y'\}$, which contradicts $Y, Y' \subseteq B$.
Hence, it follows that $X' \in \{X',Y'\} \setminus \{X,Y\}$.
Since $X, X' \subseteq A$ and $A$ is a clique of $\bar{G}$, they form a connected subgraph $\bar{G}[X \cup X']$ of $\bar{G}$.
This implies that the CJ-move from $T$ to $T'$ is also a CS-move from $T$ to $T'$.


\section{Concluding remarks}

In this paper, we showed that {\BBR} on bipartite graphs is PSPACE-complete,
resolving a previously unsettled problem of Hanaka et al.~\cite{HanakaIMMNSSV20}.
Using this hardness result, we showed that {\CCR} with two connected components is PSPACE-complete under both CJ and CS rules,
answering an open problem of Nakahata~\cite{Nakahata25} in the negative.

\paragraph{Token sliding}
In this paper, we only considered the token-jumping model for the biclique reconfiguration problem.
Indeed, the token-sliding model is not meaningful for {\BBR} on bipartite graphs, 
since any TS-move on a bipartite graph maps a $(p,p)$-biclique to a set with $p-1$ vertices on one side and $p+1$ vertices on the other side.

For general graphs, we can see that the proof of Hanaka et al.~\cite[Theorem~12]{HanakaIMMNSSV20}
implies that {\BBR} under the token-sliding model is PSPACE-complete.
Although their theorem is stated for the token-jumping model, its proof works for the token-sliding model as well.

\paragraph{Flexible variant}
One may also consider the \emph{flexible} version of the biclique reconfiguration problem
that takes as input a graph $G$, a $(p,q)$-biclique $S$, and a $(p',q')$-biclique $S'$, where $(p,q)$ and $(p',q')$ are not necessarily the same,
and asks if there is a TJ-sequence (or a TS-sequence) of (any) bicliques from $S$ to $S'$.
For general graphs, we can use the proof of Hanaka et al.~\cite[Theorem~12]{HanakaIMMNSSV20} again.
Although their proof is designed for the case where the shape of allowed bicliques is fixed, it actually works without such a restriction.

For bipartite graphs, we may show that it is at least NP-hard under TJ via a reduction from {\ISR} under TJ on bipartite graphs~\cite{LokshtanovM19},
because \emph{almost} all bicliques of a bipartite graph $H = (A,B;E)$ are independent sets of its bipartite complement $H'$:
the only exceptions are bicliques $S$ of $H$ that are indeed independent sets $K_{|S|, 0}$ of $H$ that intersect both $A$ and $B$;
such independent-set bicliques of $H$ form bicliques $K_{|S \cap A|, |S \cap B|}$ of $H'$.
In a reduction, we can avoid such exceptional bicliques by adding to both sides of the graph
a large number of vertices adjacent to all vertices in the other side
and including them into both initial and target sets.
For the token-sliding model on bipartite graphs, we did not find a direct connection from the flexible variant to known problems.


\bibliographystyle{plainurl}
\bibliography{ref}

\end{document}